\documentclass[conference]{IEEEtran}
\IEEEoverridecommandlockouts
\usepackage{cite}
\usepackage{amsmath,amssymb,amsfonts}
\usepackage{graphicx}
\usepackage{textcomp}
\usepackage{xcolor}
\def\BibTeX{{\rm B\kern-.05em{\sc i\kern-.025em b}\kern-.08em
    T\kern-.1667em\lower.7ex\hbox{E}\kern-.125emX}}

\usepackage{enumitem}
\usepackage{graphicx}
\usepackage{booktabs}
\usepackage[table]{xcolor}

\usepackage[ruled,vlined,linesnumbered]{algorithm2e}
\usepackage{threeparttable}
\usepackage{algpseudocode}
\usepackage{tikz}
\usepackage{pifont}
\usepackage{amsthm}
\usepackage{cleveref}
\usetikzlibrary{arrows.meta, positioning, shapes.geometric, calc}
\SetKwComment{tcp}{ $\triangleright$ }{} 
\SetCommentSty{textit}                   
\SetKwFor{Upon}{Upon}{do}{end}
\newcommand{\EC}{\textsf{EC}}
\newcommand{\proname}{\textsc{ACE}}
\newcommand{\PBB}{\textsf{PBB}}
\newcommand{\tone}{\mathsf{T}_1}
\newcommand{\vi}{\vec{v}_i}
\newcommand{\tj}{\mathsf{T}_j}
\newcommand{\adv}{\mathcal{A}}
\newcommand{\vobs}{\mathsf{V}_{\text{obs}}}
\newcommand{\audit}{\mathsf{Audit-or-Cast}}
\newcommand{\Tdesign}{\mathsf{T}_{\text{des}}}
\newcommand{\frest}{$f_{\mathsf{res}}(\mathbb{T})$}
\newtheorem{definition}{Definition}
\newtheorem{lemma}{Lemma}
\newtheorem{theorem}{Theorem}

\newcommand{\pki}{\textsf{PKI}~}
\newcommand{\vs}[1]{}
\newcommand{\ar}[1]{}
\newcommand{\gt}[1]{}
\begin{document}

\title{Audit-or-Cast: Enforcing Honest Elections with Privacy-Preserving Public Verification}

\author{%
\IEEEauthorblockN{Aman Rojjha\textsuperscript{$\ast$}}
\IEEEauthorblockA{\textit{CSTAR, IIIT Hyderabad}\\
Hyderabad, India\\
{\footnotesize aman.rojjha@research.iiit.ac.in}}
\and
\IEEEauthorblockN{Gaurang Tandon\textsuperscript{$\ast$}}
\IEEEauthorblockA{\textit{IREL, IIIT Hyderabad}\\
Hyderabad, India\\
{\footnotesize gaurang.tandon@alumni.iiit.ac.in}}
\and
\IEEEauthorblockN{Varul Srivastava\textsuperscript{$\ast$}}
\IEEEauthorblockA{\textit{MLL, IIIT Hyderabad}\\
Hyderabad, India\\
{\footnotesize varulsrivastava06@gmail.com}}
\and
\IEEEauthorblockN{Kannan Srinathan}
\IEEEauthorblockA{\textit{CSTAR, IIIT Hyderabad}\\
Hyderabad, India\\
{\footnotesize srinathan@iiit.ac.in}}
\thanks{$^{\ast}$Equal first authors.}
}

\maketitle


%
\IEEEpeerreviewmaketitle

\begin{abstract}
Electronic voting systems must balance public verifiability with voter privacy and coercion resistance. Existing cryptographic protocols typically achieve end-to-end verifiability by revealing vote distributions, relying on trusted clients, or enabling transferable receipts --- design choices that often compromise trust or privacy in real-world deployments.

We present \textsf{\proname}, a voting protocol that reconciles public auditability with strong privacy guarantees. The protocol combines a publicly verifiable, tally-hiding aggregation mechanism with an \emph{Audit-or-Cast} challenge that enforces cast-as-intended even under untrusted client assumptions. Tallier-side re-randomization eliminates persistent links between voters and public records, yielding information-theoretic receipt-freeness assuming at least one honest tallier.

We formalize the security of \textsf{\proname} and show that it simultaneously achieves end-to-end verifiability, publicly tally-hiding, and strong receipt-freeness without trusted clients.
\end{abstract}

\section{Introduction}

Voting is a fundamental primitive for collective decision-making in modern societies, underpinning democratic governance, corporate control, and decentralized systems. As elections scale in size and adversarial sophistication, voting protocols must simultaneously satisfy strong correctness guarantees, voter privacy, and public auditability. Achieving these properties in a single system remains a central challenge in secure distributed systems and applied cryptography \cite{cortier2016sok,kusters2011verifiability}.

\noindent\textbf{Voting as a Classical Security Problem.}
The core tension in voting systems is long-standing: ballots must remain secret, yet election outcomes must be publicly verifiable. Historically, procedural safeguards and paper-based audits were used to approximate this balance, but these mechanisms do not scale well and remain vulnerable to insider attacks and coercion. Early cryptographic work formalized these challenges, introducing notions such as receipt-freeness and coercion-resistance to prevent voters from proving how they voted, even if they are willing to do so \cite{benaloh2006simple,okamoto1997receipt}.

\noindent\textbf{Cryptographic Voting Systems.}
Modern electronic voting protocols rely on cryptographic primitives such as public-key encryption, homomorphic tallying, commitment schemes, and zero-knowledge proofs to achieve end-to-end verifiability (E2E-V)~\cite{kiayias2015e2e}. Systems such as Helios~\cite{helios}, Pr\^{e}t \`a Voter~\cite{ryan2009pret}, and Ordinos~\cite{kusters2020ordinos} demonstrate that universal verifiability can be achieved in practice. However, these systems often leak full or partial tallies, rely on trusted clients, or permit transferable receipts, thereby weakening privacy guarantees such as receipt-freeness and public tally-hiding \cite{cortier2016sok,pankova2023relations}.

\noindent\textbf{National Elections --- Persistent Trust and Auditability Gaps.}
Despite formal guarantees, deployed election systems continue to face disputes rooted in limited auditability and opaque tallying processes~\cite{reuters_georgia_election_2024,reuters_venezuela_publish_records_2024,time_pakistan_phone_suspension_2024,ap_pakistan_blocks_x_2024,reuters_india_vvpat_sc_2024,guardian_bangladesh_election_2024}. The inability for voters or third parties to independently verify election integrity without compromising ballot secrecy remains a key source of distrust. This gap highlights the need for voting protocols that provide strong public accountability while remaining secure against coercion, malicious authorities, and compromised voting devices~\cite{kusters2010accountability}. 

\noindent\textbf{Problem Statement --- Efficient, Secure and Private elections.} The most critical application of voting in a democratic world is choosing national leaders. Attacks are two-way: \emph{manipulation of election results} to sway the results in a non-democratic direction -OR- \emph{rejection of results and undermining public opinion} based on claims of forgery. We aim to address the specific challenges of national-scale elections. The system should satisfy security and privacy properties along with providing audit at each step, eliminating the possibility of `claiming' forgery without proof. Towards this, we introduce properties for privacy and security in Table~\ref{tab:properties}.

\noindent\textbf{Our Approach.}
In this work, we introduce a cryptographic voting protocol \textbf{A}udit-\textbf{C}ast \textbf{E}lection (\proname) protocol that simultaneously achieves end-to-end verifiability, publicly tally-hiding, and strong receipt-freeness. Concretely, \proname{} combines tallier-side re-randomization of \emph{perfectly hiding} Pedersen vector commitments with a Benaloh-style Audit-or-Cast challenge. Perfect hiding makes privacy information-theoretic even against all-tallier collusion, obviating shuffle proofs. Our design enforces honest behavior among mutually adversarial participants via an \emph{Audit-or-Cast} mechanism, ensuring cast-as-intended even under untrusted client assumptions. Unlike prior approaches, our protocol enables public auditability of election correctness without revealing individual votes or intermediate tallies, thereby reconciling transparency and privacy within a unified cryptographic framework. Tab.~\ref{tab:protocol-comparison} summarizes the property and communication trade-offs against prior protocols.

\begin{table}[!th]
\centering
\small
\renewcommand{\arraystretch}{1.15}
\rowcolors{2}{gray!12}{white}
\begin{tabular}{cl>{\raggedright\arraybackslash}p{3.8cm}}
\toprule
\textbf{ID} & \textbf{Property} & \textbf{Description} \\
\midrule
P1 & Public Privacy~\cite{kryvos} &
Talliers cannot link voters to candidates. \\

P2 & Publicly Tally-Hiding &
Only $\Tdesign$ learns $\mathbb{T}$; no vote shares or margins are revealed. \\

P3 & Receipt-Freeness &
Voters cannot produce transferable proof of their vote. \\

S1 & Double-Vote Inhibition &
Each voter can cast at most one valid vote. \\

S2 & Vote Immutability &
Recorded votes cannot be altered by any party. \\

S3 & Cast-as-Intended &
Voters can verify correct recording of their intent. \\

S4 & Tally-as-Intended &
Anyone can verify correct tally computation. \\

S5 & End-to-End Verifiability &
Anyone can verify $f_{\mathrm{res}}(\mathbb{T})$ without ballot disclosure. \\
\bottomrule
\end{tabular}
\caption{Privacy (P) and Security (S) properties required of the voting protocol.}
\label{tab:properties}
\end{table}

\section{Related Works}

We organize related work along four axes: game-theoretic analyses of voting, cryptographic voting protocols, differential privacy approaches, and complementary frameworks.

\subsection{Game-Theoretic Analyses of Voting}

Game theory has been used to analyze strategic behavior in voting, particularly under coercion and manipulation. Benaloh’s randomized challenge mechanism \cite{benaloh2006simple} pioneered audits as a deterrent against ballot tampering. Jamroga \cite{jamroga2023pretty} formalizes this mechanism as a Stackelberg game, showing that simple randomized audit strategies achieve near-optimal security. Our \emph{Audit-or-Cast} mechanism naturally fits this framework.

\subsection{Cryptographic Voting Protocols}

Cryptographic voting protocols aim to achieve privacy, correctness, and verifiability under adversarial conditions \cite{cortier2016sok,kusters2011verifiability}. Systems such as Helios \cite{helios}, Pr\^{e}t \`a Voter \cite{ryan2009pret}, and Demos \cite{kiayias2015e2e} provide end-to-end verifiability but rely on trusted clients or leak vote distributions. Ordinos \cite{kusters2020ordinos} achieves publicly tally-hiding elections but permits transferable receipts. Blockchain-based systems such as DeVoS \cite{devos} provide immutability at the cost of receipt-freeness. Our protocol combines public auditability, tally-hiding, and receipt-freeness without trusted clients.

\subsection{Differential Privacy and Voting}

Differential privacy has been proposed to limit information leakage from election results \cite{dwork2006dp,nissim2007smooth}. These approaches trade exact correctness for statistical privacy and generally preclude cryptographic verifiability. Moreover, DP does not address coercion or receipt-freeness. Our work avoids noise-based privacy and instead relies on protocol-level guarantees.

\subsection{Other Related Work}

Formal frameworks analyze relationships between privacy, verifiability, accountability, and coercion-resistance \cite{kusters2010accountability,pankova2023relations}. Selene \cite{ryan2016selene} improves usability via voter trackers but introduces coercion risks through persistent identifiers. In contrast, our approach relies solely on ephemeral audits and public transcripts.

While existing solutions like Helios and Kryvos achieve excellent verifiability (S5), they force a trade-off: Helios reveals the full tally (failing P2) and Kryvos lacks strong receipt-freeness (failing P3). \textsf{\proname} is the first to achieve the `Holy Trinity' of Tally-Hiding, Receipt-Freeness, and Verifiability.
A comparative study of the protocols v/s ours can be found in~\Cref{tab:protocol-comparison}

\begin{table*}[!t]
\centering
\caption{Voting Protocols vs.\ Privacy (\textbf{P}), Security (\textbf{S}), and Communication (per-voter / per-tallier)}
\label{tab:protocol-comparison}
\begin{threeparttable}
\scriptsize
\setlength{\tabcolsep}{2.5pt}
\renewcommand{\arraystretch}{1.05}
\rowcolors{2}{gray!10}{white}
\begin{tabular}{lcccccccccc}
\toprule
\textbf{Protocol} & P1 & P2 & P3 & S1 & S2 & S3 & S4 & S5 & \textbf{Voter} & \textbf{Tallier} \\
\midrule
Helios~\cite{helios}
 & \textcolor{green!60!black}{\ding{51}} & \textcolor{red!70!black}{\ding{55}} & \textcolor{red!70!black}{\ding{55}}
 & \textcolor{green!60!black}{\ding{51}} & \textcolor{green!60!black}{\ding{51}} & \textcolor{green!60!black}{\ding{51}} & \textcolor{green!60!black}{\ding{51}} & \textcolor{red!70!black}{\ding{55}}
 & $O(1)$ & $O(n_t)$ \\
Fasten~\cite{fasten}
 & \textcolor{green!60!black}{\ding{51}} & \textcolor{red!70!black}{\ding{55}} & \textcolor{red!70!black}{\ding{55}}
 & \textcolor{green!60!black}{\ding{51}} & \textcolor{green!60!black}{\ding{51}} & \textcolor{green!60!black}{\ding{51}} & \textcolor{green!60!black}{\ding{51}} & \textcolor{red!70!black}{\ding{55}}
 & $O(1)$ & $O(n_t)$ \\
DEMOS~\cite{kiayias2015e2e}
 & \textcolor{green!60!black}{\ding{51}} & \textcolor{red!70!black}{\ding{55}} & \textcolor{red!70!black}{\ding{55}}\tnote{*}
 & \textcolor{green!60!black}{\ding{51}} & \textcolor{green!60!black}{\ding{51}} & \textcolor{green!60!black}{\ding{51}} & \textcolor{green!60!black}{\ding{51}} & \textcolor{green!60!black}{\ding{51}}
 & $O(1)$ & $O(n_t)$ \\
Pr\^et \`a Voter~\cite{ryan2009pret}
 & \textcolor{green!60!black}{\ding{51}} & \textcolor{red!70!black}{\ding{55}} & \textcolor{red!70!black}{\ding{55}}\tnote{**}
 & \textcolor{green!60!black}{\ding{51}} & \textcolor{green!60!black}{\ding{51}} & \textcolor{green!60!black}{\ding{51}} & \textcolor{green!60!black}{\ding{51}} & \textcolor{red!70!black}{\ding{55}}
 & $O(1)$ & $O(k_{\mathrm{mix}})$\,+\,shuf. \\
Kryvos~\cite{kryvos}
 & \textcolor{green!60!black}{\ding{51}} & \textcolor{green!60!black}{\ding{51}} & \textcolor{red!70!black}{\ding{55}}
 & \textcolor{green!60!black}{\ding{51}} & \textcolor{green!60!black}{\ding{51}} & \textcolor{red!70!black}{\ding{55}} & \textcolor{green!60!black}{\ding{51}} & \textcolor{green!60!black}{\ding{51}}
 & $O(1)$ & $O(\mathrm{poly}(n_t))$ \\
Ordinos~\cite{kusters2020ordinos}
 & \textcolor{green!60!black}{\ding{51}} & \textcolor{green!60!black}{\ding{51}} & \textcolor{red!70!black}{\ding{55}}
 & \textcolor{green!60!black}{\ding{51}} & \textcolor{green!60!black}{\ding{51}} & \textcolor{green!60!black}{\ding{51}} & \textcolor{green!60!black}{\ding{51}} & \textcolor{green!60!black}{\ding{51}}
 & $O(1)$ & $O(\mathrm{poly}(n_t))$ \\
DeVoS~\cite{devos}
 & \textcolor{green!60!black}{\ding{51}} & \textcolor{red!70!black}{\ding{55}} & \textcolor{red!70!black}{\ding{55}}
 & \textcolor{green!60!black}{\ding{51}} & \textcolor{green!60!black}{\ding{51}} & \textcolor{green!60!black}{\ding{51}} & \textcolor{green!60!black}{\ding{51}} & \textcolor{red!70!black}{\ding{55}}
 & $O(1)$ & $O(n_t)$ \\
Selene~\cite{ryan2016selene}
 & \textcolor{green!60!black}{\ding{51}} & \textcolor{red!70!black}{\ding{55}} & \textcolor{red!70!black}{\ding{55}}
 & \textcolor{green!60!black}{\ding{51}} & \textcolor{green!60!black}{\ding{51}} & \textcolor{green!60!black}{\ding{51}} & \textcolor{green!60!black}{\ding{51}} & \textcolor{green!60!black}{\ding{51}}
 & $O(1)$ & $O(n_t)$ \\
\midrule
\rowcolor{green!15}
\textbf{\textsf{\proname} (Ours)}
 & \textbf{\textcolor{green!70!black}{\ding{51}}} & \textbf{\textcolor{green!70!black}{\ding{51}}} & \textbf{\textcolor{green!70!black}{\ding{51}}}
 & \textbf{\textcolor{green!70!black}{\ding{51}}} & \textbf{\textcolor{green!70!black}{\ding{51}}} & \textbf{\textcolor{green!70!black}{\ding{51}}} & \textbf{\textcolor{green!70!black}{\ding{51}}} & \textbf{\textcolor{green!70!black}{\ding{51}}}
 & $\mathbf{O(k\,n_t)}$ & $\mathbf{O(n_t)}$ \\
\bottomrule
\end{tabular}
\begin{tablenotes}
\scriptsize
\item[*] DEMOS fails our receipt-freeness definition if the voter shares code sheet + session transcript with a coercer.
\item[**] Pr\^et \`a Voter claims receipt-freeness but is vulnerable to ``pattern coercion'' via pre-cast ballot photography.
\end{tablenotes}
\end{threeparttable}
\end{table*}

\section{Preliminaries}\label{sec:preliminaries}

We briefly recall the cryptographic primitives and system abstractions
used throughout the protocol.
We treat all constructions as black boxes, assuming security with a computational security parameter $\ell$.

\begin{table}[htbp]
\caption{Notation and Terminology}
\label{tab:terminology}
\centering
\begin{tabular}{p{0.25\linewidth}p{0.70\linewidth}}
\hline
\textbf{Notation} & \textbf{Description} \\
\hline
$\mathcal{V}$ & Voter set, where $n = |\mathcal{V}|$ \\
\hline
$C$ & Choice set, where $C \subseteq (\mathbb{F_q})^{n_{\text{choices}}}$ where $n_{\text{choices}}$ is the number of candidates and $q$ is the maximum number of votes allowed for a single candidate \\
\hline
Election Commission (EC) & Centralized entity for setting up election. Special meaning for \textit{publicly tally-hiding} in elections. \\
\hline
Talliers $\mathsf{T}_j$ & Tallier set $\mathcal{T}$ (size $n_t$) responsible for receiving voter shares $v_i^{(j)}$ from the voters \\
\hline
$\mathsf{T}_{\text{des}}$ & Designated tallier for accumulating the final tally \\
\hline
$\vec{v}_i^{(j)}$ & Vote share submitted by $i$th voter to tallier $\tj$ \\
\hline
$\vec{v}_i$ & Vote finalized by voter $i$ \\
\hline
$c_i^{(k)}$ & Pedersen Vector Commitment\cite{pedersonvector} over vote share $\vec{v}_i^{(k)}$ under the group generators $g$ and $h$ as specified on $\PBB$\\
\hline
$r_i^{(j)}$ & Randomness used by voter $i$ for computing commitment $c_i^{(j)}$ for it's $j$th vote-share $\vec{v}_i^{(j)}$ \\
\hline
$\tilde{c}_i^{(k)}$ & Pedersen Commitment\cite{pedersoncommitment} over vote share $\vec{v}_i^{(k)}$ under the group generators $g$ and $h$ as specified on $\PBB$ \\
\hline
$\tilde{r}_i^{(j)}$ & Randomness used by tallier $\tj$ for commiting to the voter $i$'s vote-share commitment $c_i^{(j)}$\\
\hline
$\mathbb{T}$ & Final tally after voting period ends \\
\hline
$f_{\text{res}}$ & Efficient (polynomial) function that computes the  final result from tally $\mathbb{T}$ \\
\hline
$\PBB$ & Public Bulletin Board\cite{bulletinboards} providing robust storage for election transcript \\
\hline
Auditors & Any public entity (including voters themselves) which can verify the properties of the protocol from the transcripts available on $\PBB$ \cite{bulletinboards} \\
\hline
\end{tabular}
\end{table}

\subsection{Public Bulletin Boards}\label{sec:pbb}

A public bulletin board (\PBB) is an append-only, publicly readable
broadcast channel with the following properties \cite{bulletinboards}:

\begin{itemize}
    \item \textbf{Persistence:} Once a message is posted, it cannot be removed or altered.
    \item \textbf{Public Verifiability:} Any party can independently read and verify all entries.
    \item \textbf{Global Consistency:} All honest parties observe the same board contents.
\end{itemize}

The \PBB\ may be instantiated using a permissionless or permissioned
blockchain, a replicated append-only log, or a consensus-backed ledger.
We treat the \PBB\ as an ideal abstraction.

\subsection{Pedersen Vector Commitment}\label{def:pedersenvector}
Let $pp_{\mathsf{CommVec}} = (\mathbb{G}, q, h, \mathbf{g})$ be the public parameters where $\mathbf{g} = (g_1, \dots, g_n) \in \mathbb{G}^n$ and $h \in \mathbb{G}$. We define the algorithms as follows:

\begin{enumerate}
    \item $\mathsf{CommVec}(\vec{v}; r)$: Computes a commitment $c$ to vector $\vec{v} \in \mathbb{F}_q^n$ with randomness $r \in \mathbb{F}_q$:
    \begin{equation}
        c = h^r \cdot \prod_{k=1}^{n} g_k^{v[k]}
    \end{equation}

    \item $\mathsf{ReRand}(c; r')$: Updates commitment $c$ with fresh randomness $r' \in \mathbb{F}_q$ to produce $c'$:
    \begin{equation}
        c' = c \cdot h^{r'} = \mathsf{CommVec}(\vec{v}; r + r')
    \end{equation}

    \item $\mathsf{DeRand}(c'; r')$: Reverts $c'$ using factor $r'$ to retrieve original commitment $c$:
    \begin{equation}
        c = c' \cdot h^{-r'} = \mathsf{CommVec}(\vec{v}; (r+r') - r')
    \end{equation}
\end{enumerate}

\noindent \textbf{Relationship:} The functions satisfy the consistency relation $\mathsf{DeRand}(\mathsf{ReRand}(c; r'); r') = c$.

\subsection{\pki}

We assume \textit{Public  Key Infrastructure}, i.e., all the involved entities (voter, $\tj$, auditors) have a consistent view of each other's public keys.

\noindent \textbf{Digital Signatures} Schemes are a $3$-tuple $(\mathsf{Gen, Sign, Vrfy})$ ($(\mathsf{Gen(\cdot), Sign_{sk}(m), Vrfy_{pk}(\sigma,m)})$) which can verify the integrity of a message w.r.t a public key $pk$, i.e., if the message $m$ was sent by the entity holding the secret key $sk$ corresponding to $pk$ with which message was signed.

\subsection{Non-Interactive Zero Knowledge Proofs}\label{def:nizkp}
Let $\mathcal{R}$ be a relation generator for NP statements, defined as pairs $(x, w)$ where $x$ is the public instance and $w$ is the private witness. A NIZKP scheme $\Pi$ consists of three efficient algorithms:

\begin{enumerate}
    \item $\Pi.\mathsf{Setup}(\mathcal{R}) \to (\sigma_{pk}, \sigma_{vk})$: A trusted setup that takes the relation $\mathcal{R}$ (represented as an arithmetic circuit) and outputs a proving key $\sigma_{pk}$ and a verification key $\sigma_{vk}$.
    
    \item $\Pi.\mathsf{Prove}(\sigma_{pk}, x, w) \to \pi$: Given the proving key, a public statement $x$, and a valid witness $w$ such that $(x, w) \in \mathcal{R}$, outputs a succinct proof $\pi$ consisting of three group elements.
    
    \item $\Pi.\mathsf{Verify}(\sigma_{vk}, x, \pi) \to \{0, 1\}$: Given the verification key, statement $x$, and proof $\pi$, outputs $1$ (Accept) if the proof is valid, and $0$ (Reject) otherwise.
\end{enumerate}

\noindent The protocol satisfies \textit{Completeness}, \textit{Computational Knowledge Soundness}, and \textit{Perfect Zero-Knowledge}. We use \cite{groth16} proof system.

\subsection{Sequential Audit-or-Cast Interactive Proof}\label{seq-audit-cast}

The $\mathsf{Audit-or-Cast}$ phase \Cref{audit-phase} functions as a \textbf{sequential Cut-and-Choose}\cite{cut-and-choose} interactive proof system between a Prover (Tallier $\tj$) and a Verifier (Voter $v_i$), just like \cite{benaloh-challenge}. The goal is to prove the statement $\mathcal{S}$ that the published commitment $\tilde{c}$ is a valid re-randomization of the voter's input $c$, i.e., $\mathcal{S}:\exists \tilde{r} \in \mathbb{F}_q \text{ s.t. } \tilde{c} = c \cdot h^{\tilde{r}}$.

The protocol requires the voter to perform $k$ sequential executions, auditing the first $k-1$ attempts and casting the $k$-th, where $k \in \mathbb{N}$ is arbitrary (unknown to the adversary $\adv$).

\begin{definition}[$\mathsf{PoIO}$] \label{PoIO}
A Proof of Incorrect Opening is a tuple $\pi_{\bot} = (i, j, \tilde{c},\tilde{r})$ where:
\begin{itemize}
    \item $i, j$ are the voter and tallier indices.
    \item $\tilde{c}$ is the re-randomized commitment published on $\PBB$.
    \item $\tilde{r}$ is the opening scalar sent by $\tj$ to $v_i$.
\end{itemize}
The proof $\pi_{\bot}$ is \textbf{valid} if and only if    $\tilde{c} \neq c_i^{(j)} \cdot h^{\tilde{r}}$.
\end{definition}

\begin{definition}[Proof of Incorrect Opening (NIZK variant)]\label{def:poio-nizkp}
If aggregate verification fails (i.e., $\pi_{\text{vote}_i}$ rejects against the reconstructed $c_i=\prod_{j=1}^{n_t} c_i^{(j)}$), the talliers publish
\[
\mathsf{PoIO}_{\text{NIZKP}} = \bigl(i,\;\pi_{\text{vote}_i},\;\{\langle c_i^{(j)}, \sigma_{v_i}(c_i^{(j)})\rangle\}_{j=1}^{n_t}\bigr).
\]
Blame is deterministic: any $\tj$ unable to supply $\sigma_{v_i}(c_i^{(j)})$ for its published share is evicted; otherwise $v_i$ is classified as malicious. Under a $(1,n_t)$ model, a single honest tallier suffices to surface this proof.
\end{definition}

\section{Model}\label{sec:model}
We describe the system/entities, threat assumptions, and formal properties required of the protocol. 
Our presentation follows standard e-voting modeling frameworks~\cite{kusters2010accountability,kusters2011verifiability}.

\subsection{Participants}\label{sec:participants}

The protocol runs among the following entities:

\begin{itemize}
    \item \textbf{Election Commission (\EC):}
    \EC\ executes a \emph{one-time} pre-election setup (CRS via a distributed MPC ceremony, or a CRS-free transparent NIZK such as Bulletproofs~\cite{bulletproofs}) but holds no runtime authority.

    \item \textbf{Voters ($v_i \in \mathcal{V}$):} 
    Each voter selects a choice $\vi \in C \subseteq (\mathbb{F}_q)^{n_{\text{choices}}}$ and participates
    in the Audit-or-Cast protocol to submit their vote.

    \item \textbf{Talliers ($\tj \in \mathcal{T}$):} 
    A set of mutually distrustful parties responsible for receiving blinded vote shares, validating openings, and contributing to tally reconstruction.

    \item \textbf{Public Bulletin Board (\PBB):}
    An append-only, publicly readable, and persistent broadcast channel~\cite{bulletinboards}
    that stores commitments, proofs, and final results. \PBB\ runs as a smart contract on a permissionless blockchain: it stores transcripts, enforces phase transitions via block heights, and rejects duplicate submissions per voter ID. This eliminates EC-led DoS/censorship and closes the double-vote channel at the ledger layer.

    \item \textbf{Designated Tallier ($\Tdesign$):}
    A publicly known entity responsible for reconstructing the final tally
    and producing a zero-knowledge proof of correctness.

    \item \textbf{Judge (\textsf{J}):}
    A virtual verification entity that, given only public information,
    accepts or rejects a protocol execution.

    \item \textbf{Scheduler (\textsf{S}):}
    A virtual process responsible for coordinating protocol phases
    and modeling the behavior of \EC.
\end{itemize}

All parties possess public-key identities established via a PKI,
and all messages are digitally signed.

\subsection{Threat Model}\label{sec-model}

We model the adversary $\mathcal{A}$ as a static, active (malicious) entity capable of corrupting a subset of protocol participants. We assume mutually authenticated channels \cite{mtls-rfc} between all parties.

\begin{itemize}
    \item \textbf{Voter Independence:} 
    Let $\mathcal{V}_{corr} \subset \mathcal{V}$ be the set of voters corrupted by $\mathcal{A}$. We assume that malicious voters are computationally independent; they do not share their private witnesses (randomness $r_i$ or shares $\vec{v}_i^{(j)}$) to construct a joint proof. 
    \item \textbf{Tallier Thresholds:}
    \begin{itemize}
        \item \textit{Privacy:} Since the protocol uses $(n, n)$-threshold secret sharing, voter privacy is preserved as long as \textit{at least one tallier remains honest} and refuses to reveal their received shares.
        \item \textit{Integrity (Anti-Censorship):} We assume that a tallier $\tj$ does not collude with the entire remaining set of talliers $\mathcal{T} \setminus \{\tj\}$ to suppress valid votes.
    \end{itemize}
    Substituting a $(t, n_t)$ Pedersen verifiable secret sharing~\cite{pedersoncommitment} preserves the homomorphism while tolerating offline talliers.
\end{itemize}

\subsection{Protocol properties}\label{protocol-properties}

We discuss the intended privacy (\textbf{P}) and security (\textbf{S}) properties informally below:

\begin{itemize}
    \item[P1] \textbf{Public Privacy}: Any proper subset of ${\tj} \subset \mathcal{T}$ should not be able to verify if $v_i$ corresponds to a candidate $c_i$.
    \item[P2] \textbf{Publicly Tally-Hiding}: $\Tdesign$ knows final tally $\mathbb{T}$ but no one else does. 
    Public disclosure of exact vote counts may deter candidates fearing reputational damage from wide margins of defeat. On the other hand, if a candidate widely perceived as popular wins by only a narrow margin (say, 51\% to 49\%), it undermines their perceived mandate and weakens their position. Therefore, public exposure can damage community trust and political stability in either case.
    \item[P3] \textbf{Receipt-Freeness}: The voter should not be able to obtain information proving their vote\cite{bernhard2011helios}. This is to avoid voter coercion by third party or vote-selling by the voter himself. Also note that receipt-freeness inherently implies coercion-resistance.
    \item[S1] \textbf{Double-Vote Inhibition}: A voter should not be able to vote multiple times.
    \item[S2] \textbf{Vote Immutability}: Any party $i$ should not be able to change $v_i$ (including the voter himself) once the vote has been casted/finalized.
    \item[S3] \textbf{Cast-as-intended}: A voter should be able to verify if his vote $v_i$ has been recorded correctly.
     \item[S4] \textbf{Tally-as-intended}: Any voter should be able to verify the final tally $\mathbb{T}$ is correct.
    \item[S5] \textbf{End-to-end Verifiability}: When voting is over and winner (or set of winners) is declared, then any person should be able to verify the results $f_{res}(\mathbb{T})$.
\end{itemize}

\begin{algorithm}[t]
\SetAlgoLined
\DontPrintSemicolon
\caption{Vote Generation and Audit-or-Cast (Voter $v_i$)}\label{alg:vote-gen}
\KwIn{Vote $v_i \in \mathcal{C}$, Talliers $\mathcal{T} = \{\mathsf{T}_1, \dots, \mathsf{T}_{n_t}\}$, Public Board $\PBB$}
\KwOut{Committed vote on $\PBB$ or $\bot$}

\tcp{1. Vote Generation \& Sharding}
Split $v_i$ into full-threshold shares $\{\vec{v}_i^{(j)}\}_{j=1}^{n_t}$ s.t. $\sum \vec{v}_i^{(j)} = v_i \mod q$\;
Sample randomness $r_i^{(j)} \leftarrow \mathbb{Z}_q$ for all $j \in [n_t]$\;
Compute commitments $c_i^{(j)} \leftarrow \mathsf{CommVec}(\vec{v}_i^{(j)}; r_i^{(j)})$\;
Generate NIZKP $\pi_{\text{vote}_i}$ proving well-formedness of $v_i$ and $\{c_i^{(j)}\}$\;

\tcp{2. Blinded Submission}
\For{$j \in [n_t]$}{
    Send $(i, c_i^{(j)})$ to $\mathsf{T}_j$ via private channel\;
}
Verify presence of re-randomized commitments $\{(i, \tilde{c}_i^{(j)})\}$ on $\PBB$\;

\tcp{3. Audit-or-Cast Decision}
Select $b \leftarrow \{0, 1\}$ \tcp*[r]{0: Audit, 1: Cast}
\If{$b = 0$ (\textsf{AUDIT})}{
    Send $\mathsf{AUDIT}$ request to all $\mathsf{T}_j \in \mathcal{T}$\;
    Receive blinding factors $\{\tilde{r}_i^{(j)}\}$ from talliers\;
    \For{$j \in [n_t]$}{
        \If{$c_i^{(j)} \neq \mathsf{DeRand}(\tilde{c}_i^{(j)}; \tilde{r}_i^{(j)})$}{
            Construct and publish \textbf{PoIO} on $\PBB$\;
            \textbf{return} $\bot$ (Restart protocol)\;
        }
    }
    Wait for $\PBB$ to discard submission, then \textbf{goto} Step 1\;
}
\Else(\textsf{CAST}){
    Send $\mathsf{CAST}$ request to all $\mathsf{T}_j \in \mathcal{T}$\;
    Wait for $\mathsf{CAST}$ finalization on $\PBB$\;
    \tcp{4. Opening}
    \For{$j \in [n_t]$}{
        Send opening $(\vec{v}_i^{(j)}, r_i^{(j)})$ to $\mathsf{T}_j$ via \textbf{private channel}\;
    }
}
\end{algorithm}

\subsection{Voting phase}\label{voting-phase}

\subsubsection{Vote Generation}\label{phase:vote-gen}
This phase represents the local computational steps performed by Voter $i$. Voter $i$ selects their vote $\vi \in \mathcal{C}$ and do the following:
\begin{enumerate}
    \item \textbf{Sharding:} The voter splits the vote into fully-threshold vote shares $\vec{v}_i^{(j)}$ such that $\sum_{j=1}^{n_t}\vec{v}_i^{(j)} \mod q = \vi$.
    \item \textbf{Commitment:} The voter creates Pedersen Vector Commitments \cite{pedersonvector} for each share:
    $$c_i^{(j)} =  \mathsf{CommVec}(\vec{v}_i^{(j)}; r_i^{(j)})\quad \forall \ j \in [n_t]$$
    \item \textbf{Proof Generation:} Voter $i$ generates a NIZKP $\pi_{\text{vote}_i}$ attesting to the well-formedness of $\vi \in \mathcal{C}$ and that the set of commitments $\{c_i^{(j)}\}_{j=1}^{n_t}$ reconstructs to this valid vote.
\end{enumerate}  




\subsubsection{Blinded Submission} \label{phase:blind-submission}
\begin{itemize}
    \item Voter $i$ sends the commitment $c_i^{(j)}$ to each Tallier $\tj$ over a \textbf{private channel}.
    \item Tallier $\tj$ generates a random blinding factor $\tilde{r}_i^{(j)}$ and \textbf{re-randomizes} the commitment \cite{pedersoncommitment}:
    $$\tilde{c}_i^{(j)} = \mathsf{ReRand}(c_i^{(j)};\tilde{r}_i^{(j)}) = \mathsf{CommVec}(\vi^{(j)}; r_i^{(j)}+\tilde{r}_i^{(j)})$$
    \item $\tj$ publishes the re-randomized commitment pair $(i, \tilde{c}_i^{(j)})$ on the $\PBB$.
\end{itemize}

\subsubsection{\textsf{Audit-or-Cast} Mechanism}\label{audit-phase}

Once the re-randomized commitments appear on the $\PBB$, Voter $i$ verifies their presence. The voter must then choose to strictly execute \textbf{one} of the following actions:

\noindent \textbf{Option A: AUDIT (Challenge).} The voter challenges the talliers to prove they stored the commitments honestly.
\begin{enumerate}
    \item Voter sends an $\mathsf{AUDIT}$ request to all $\tj \in \mathcal{T}$.
    \item Each $\tj$ responds by privately sending the blinding factor $\tilde{r}_i^{(j)}$ to Voter $i$.
    \item Voter checks if $c_i^{(j)} = \mathsf{DeRand}(\tilde{c}_i^{(j)}; \tilde{r}_i^{(j)})$ where $\tilde{c}_i^{(j)}$ exists on $\PBB$.
    \item \textit{Outcome:} The vote $\vi$ and its associated commitments on $\PBB$ become \textbf{invalid}. The talliers discard the submission. The Voter \textbf{must return to Phase \ref{phase:vote-gen}} (Vote Generation) to generate fresh randomness and shares before re-submitting.\footnotemark{}
    \end{enumerate}

\noindent \textbf{Option B: CAST.} If the voter is satisfied (or chooses not to audit), they finalize the vote.
\begin{enumerate}
    \item Voter sends a $\mathsf{CAST}$ request to all $\tj$.
    \item This signal \textit{commits} the vote-shares for voter $i$ as final on the $\PBB$. No further auditing of this specific instance is permitted.
    \end{enumerate}
    
\footnotetext{Security Note: If a voter detects a mismatch during Audit (meaning $\tj$ maliciously altered the commitment), the voter constructs a Proof of Incorrect Opening (PoIO) using the digitally signed messages from the submission step. This is published to $\PBB$, triggering the removal of the malicious $\tj$.}
    

\subsubsection{\textbf{Opening and Verification}}\label{phase:opening}
This phase executes only after the $\mathsf{CAST}$ signal is finalized for voter $i$ on the $\PBB$.

\begin{enumerate}
    \item \textbf{Opening:} Voter $i$ sends the openings $(\vi^{(j)}, r_i^{(j)})$ to the respective $\tj\in\mathcal{T}$ over a \textbf{private channel}.
    \item \textbf{Tallier Verification:} Upon receiving share $c_i^{(j)}$, $\tj$ stores the voter-supplied signature $\sigma_{v_i}(c_i^{(j)})$ and validates the opening against the stored commitment.
    \item \textbf{Synchronization:} All $\tj \in \mathcal{T}$ synchronize the commitments to reconstruct the aggregated commitment $c_i = \prod_{j=1}^{n_t} c_i^{(j)}$.
    \item \textbf{Proof Validation:} Each $\tj$ verifies the NIZKP $\pi_{\text{vote}_i}$. Upon success, the talliers generate a threshold signature $\sigma_{\mathcal{T}}$ attesting to the validity of the vote and publish it to $\PBB$.
\end{enumerate}



\subsubsection{Validity Tests}\label{validity-tests}
\begin{enumerate}
    \item In the $\audit$ phase,  if the tallier $\tj$ sends incorrect opening $\hat{r}_i^{(j)}$, voter $v_i$ creates a \textit{Proof of Incorrect Opening (PoIO)} (discussed more in \Cref{PoIO}) of $j$th vote-share commitment along with digitally-signed opening $\hat{r}_i^{(j)}$ and commitment $\tilde{c}_i^{(j)}$ it received and publishes on $\PBB$. 
    \item If $\tj$ finds $\pi_{\text{vote}_i}$ invalid, vote-shares for voter $i$ are considered \textit{invalid} and discarded from the final tally after publishing $\mathsf{PoIO}_{\text{NIZKP}}$ (\Cref{def:poio-nizkp}) along with the stored signatures $\sigma_{v_i}(c_i^{(j)})$ and aggregate commitment $c_i$; blame assignment follows the deterministic eviction rule in \Cref{def:poio-nizkp}.
    
    Upon receiving a ballot, $\PBB$ checks the current state to verify that the associated voter has not previously committed a vote. If a prior entry exists, $\PBB$ aborts the operation and rejects the submission; otherwise, the vote is committed.
    \item If \textbf{taillier verification} \Cref{phase:opening} fails for a vote-share opening $(\vec{v}_i^{(j)}, r_i^{(j)})$ at tallier $\tj$, $\tj$ publishes a  \textit{Proof of Incorrect Opening (PoIO)} along with digitally-signed messages it received $(\vec{v}_i^{(j)}, r_i^{(j)})$ on the $\PBB$. This proof consists of the digitally signed commitment $c_i^{(j)}$ (from the Voting Phase) and the conflicting opening received. 
\end{enumerate}

\textbf{Voting} phase ends after $\textsf{CAST}$ signal is finalized for all the voters $v_i \in \mathcal{V}$ on the $\PBB$.

\begin{algorithm}[t]
\SetAlgoLined
\DontPrintSemicolon
\caption{Tallier Operations (Tallier $\mathsf{T}_j$)}\label{alg:tallier-ops}
\KwIn{Incoming commitments $c_i^{(j)}$, signals from $v_i \forall v_i \in \mathcal{V}$}

\tcp{Phase: Blinded Submission}
\Upon{receiving $c_i^{(j)}$ from $v_i$}{
    Sample blinding factor $\tilde{r}_i^{(j)} \leftarrow \mathbb{Z}_q$\;
    Compute $\tilde{c}_i^{(j)} \leftarrow \mathsf{ReRand}(c_i^{(j)};\tilde{r}_i^{(j)})$\;
    Publish $(i, \tilde{c}_i^{(j)})$ to $\PBB$\;
}

\tcp{Phase: Audit-or-Cast}
\Upon{receiving $\mathsf{AUDIT}$ from $v_i$}{
    Send $\tilde{r}_i^{(j)}$ to $v_i$ over private channel\;
    Discard stored vote-share commitment $\tilde{c}_i^{(j)}$ for $v_i$ from $\PBB$\;
}
\Upon{receiving $\mathsf{CAST}$ from $v_i$}{
    Mark vote-share commitment $\tilde{c}_i^{(j)}$ as \textbf{finalized} on $\PBB$\;
    Receive opening $(\vec{v}_i^{(j)}, r_i^{(j)})$ from $v_i$\;
    \If{$\mathsf{CommVec}(\vec{v}_i^{(j)}; r_i^{(j)}) \neq c_i^{(j)}$}{
        Publish \textbf{PoIO} on $\PBB$ and discard vote\;
    }
    Store verified share $\vec{v}_i^{(j)}$\;

}
\tcp{Synchronization \& Proof Validation}
Exchange commitments $c_i^{(j)}$ with other talliers $\mathcal{T} \setminus \{\tj\}$\;
Reconstruct aggregated commitment $c_i \leftarrow \prod_{k=1}^{n_t} c_i^{(k)}$\;
\If{$\Pi_{\mathcal{R}_{\text{vote}}}.\mathsf{Verify}(\pi_{\text{vote}_i}, c_i) = \text{valid}$}{
    Generate threshold signature share $\sigma_{\mathcal{T}}^{(j)}$\;
    Cooperatively aggregate $\sigma_{\mathcal{T}}$ and publish to $\PBB$\;
} \Else {
    Flag voter $i$ as invalid on $\PBB$\;
    Reset voter $i$'s vote-share $\vi^{(j)} = \vec{0}$\;
}
\tcp{Phase: Tally Aggregation (Post-Voting)}
Compute aggregate share $\vec{v}_{\perp}^{(j)} = \sum_{i=1}^{n} \vec{v}_i^{(j)}$\;
Compute aggregate randomness $r_{\perp}^{(j)} = \sum_{i=1}^n r_i^{(j)}$ and $\tilde{r}_{\perp}^{(j)} = \sum_{i=1}^n \tilde{r}_{i}^{(j)}$\;
Send tuple $(\vec{v}_{\perp}^{(j)}, r_{\perp}^{(j)}, \tilde{r}_{\perp}^{(j)})$ to $\Tdesign$\;
\end{algorithm}

\subsection{Tally Aggregation}\label{tally-phase}
For simplicity, we assume that all vote-shares were valid. For each $j \in [n_t]$, Tallier $\tj$ performs the following:
\begin{enumerate}
    \item \textbf{Share Aggregation:} $\tj$ sums the valid vote-shares and corresponding randomness for all $n$ voters:
    $$ \vec{v}_{\perp}^{(j)} = \sum_{i=1}^{n} \vec{v}_i^{(j)}, \quad r_{\perp}^{(j)} = \sum_{i=1}^n r_i^{(j)}, \quad \tilde{r}_{\perp}^{(j)} = \sum_{i=1}^n \tilde{r}_{i}^{(j)} $$
    
    \item \textbf{Transmission:} $\tj$ sends the tuple $(\vec{v}_{\perp}^{(j)}, r_{\perp}^{(j)}, \tilde{r}_{\perp}^{(j)})$ to the designated entity $\Tdesign$ over an authenticated channel.
\end{enumerate}

\subsection{Result Phase}\label{phase:result}
$\Tdesign$ is responsible for reconstructing the final tally from the aggregates provided by the talliers. 

\subsubsection{Tallier Consistency Check}\label{step:consistency}
Before reconstruction, $\Tdesign$ must verify that the aggregate shares received from each $\tj$ correspond to the sum of commitments published on $\PBB$.

For every tallier $\tj \in \mathcal{T}$:
\begin{enumerate}
    \item $\Tdesign$ computes the expected aggregate commitment from the public board:
    $$ \tilde{c}_{\perp}^{(j)} = \prod_{i=1}^{n} \tilde{c}_i^{(j)} $$
    \item $\Tdesign$ verifies\footnote{If this verification fails for any $\tj$, the tallier is proven malicious. A \textit{PoIO} is published on $\PBB$, the tallier $\tj$'s shares are discarded, and voters must re-submit shares to a fallback tallier $\mathsf{T}_k$.} that the received shares are a valid opening for this aggregate commitment:
    $$ \tilde{c}_{\perp}^{(j)} \stackrel{?}{=} \mathsf{CommVec}(\vec{v}_{\perp}^{(j)}; \ r_{\perp}^{(j)} + \tilde{r}_{\perp}^{(j)}) $$
\end{enumerate}

\subsubsection{Global Reconstruction}
Upon successful verification of all talliers, $\Tdesign$ computes the final election tally $\mathbb{T}$ by summing the aggregate shares:
$$ \mathbb{T} = \sum_{j=1}^{n_t} \vec{v}_{\perp}^{(j)} $$

\subsubsection{Result Publication and Proof}
To finalize the election, $\Tdesign$ publishes the result and proves its correctness relative to the public commitments.

\begin{enumerate}[leftmargin=*]
    \item \textbf{Derandomization:} $\Tdesign$ aggregates the blinding factors to isolate the global randomness:
    $$ \tilde{r}_{\perp} = \sum_{j=1}^{n_t} \tilde{r}_{\perp}^{(j)} $$
    Note that the global commitment to the result $c_{\perp}$ can be derived publicly as:
    $$ c_{\perp} = \mathsf{DeRand}\left( \prod_{j=1}^{n_t} \tilde{c}_{\perp}^{(j)}; \ \tilde{r}_{\perp} \right) $$
    
    \item \textbf{Proof Generation:} $\Tdesign$ generates a NIZKP $\pi_{\text{res}}$ for the relation $\mathcal{R}_{\text{res}}$, proving that the final tally $\mathbb{T}$ used for calculating the plain-text result $f_{\text{res}}(\mathbb{T})$ corresponds to the committed value in $c_{\perp}$.
    
    \item \textbf{Publication:} $\Tdesign$ publishes the tuple $(f_{\text{res}}(\mathbb{T}), \tilde{r}_{\perp}, \pi_{\text{res}})$ on $\PBB$.
\end{enumerate}

\begin{algorithm}[t]
\SetAlgoLined
\DontPrintSemicolon
\caption{Result Reconstruction ($\Tdesign$)}\label{alg:result}
\KwIn{Aggregates from Talliers $\{ (\vec{v}_{\perp}^{(j)}, r_{\perp}^{(j)}, \tilde{r}_{\perp}^{(j)}) \}_{j=1}^{n_t}$, $\PBB$}
\KwOut{Final Tally $\mathbb{T}$ and Proof $\pi_{\text{res}}$}

\For{$j \in [n_t]$}{
    Fetch public aggregate $\tilde{c}_{\perp}^{(j)} = \prod_{i=1}^{n} \tilde{c}_i^{(j)}$ from $\PBB$\;
    \tcp{Consistency Check}
    \If{$\tilde{c}_{\perp}^{(j)} \neq \mathsf{CommVec}(\vec{v}_{\perp}^{(j)}; r_{\perp}^{(j)} + \tilde{r}_{\perp}^{(j)})$}{
        Identify malicious $\mathsf{T}_j$ via PoIO and \textbf{abort}\;
    }
}

\tcp{Global Reconstruction}
Compute Final Tally $\mathbb{T} = \sum_{j=1}^{n_t} \vec{v}_{\perp}^{(j)}$\;
Compute Global Blinding $\tilde{r}_{\perp} = \sum_{j=1}^{n_t} \tilde{r}_{\perp}^{(j)}$\;
Reconstruct global commitment $c_{\perp} \leftarrow \mathsf{DeRand}(\prod_{j} \tilde{c}_{\perp}^{(j)}; \tilde{r}_{\perp})$\;

\tcp{Proof Generation}
Generate NIZKP $\pi_{\text{res}}$ s.t. $\mathcal{R}_{\text{res}}(\mathbb{T}, c_{\perp})$ holds\;
Publish $(f_{\text{res}}(\mathbb{T}), \tilde{r}_{\perp}, \pi_{\text{res}})$ on $\PBB$\;
\end{algorithm}

\subsection{Verification Phase}
Auditors can verify the \textbf{S} and \textbf{P} properties by verifying the final accumulation of ``valid'' votes on the $\PBB$, the final NIZKP was computed over witness with tally commitment  $c_{\perp} = \prod_{j=0}^{n_t} c_{\perp}^{(j)}$ and the correctness of computation of the final result \frest. We discuss the relevant proof of security and integrity in \Cref{verif-theorem}.

\section{Protocol Analysis} \label{analysis}

We consider the security model as discussed in \Cref{sec-model}. 

\subsection{Computational model}

We formally model our protocol in a general computational framework that we can leverage to analyze security \cite{kusters2010accountability} as well as privacy \cite{kusters2011verifiability} properties, used extensively in literature for verifiability \cite{kryvos,devos,pankova2023relations,kusters2020ordinos} and privacy properties\cite{devos,kryvos,kusters2020ordinos,helios}. 

 We denote $\pi = (\pi_P\mid\mid\pi_{\mathcal{A}})$ for the process where honest protocol processes $\pi_P$  alongside the adversarial processes which can control both the network and a static set of protocol participants. Auditors (anyone) can act as judge \textsf{J}, i.e., run the program $p_{\text{J}}$ of judge on public input, to \textit{verify} the protocol run. 

The goal $\gamma$ of protocol $P$ is  a set of protocol runs for which the election result corresponds to the actual voter choice, where the description of a run includes the description of the protocol, the adversary with which the protocol is run, and the random  coins used by these entities. According to \cite{kusters2010accountability}, a goal $\gamma$ is verifiable by a judge $\text{J}$ in a protocol P if and only if the probability of a judge $\text{J}$ accepting a run $r$ of protocol P in violation of the goal $\gamma$, i.e. $r \notin \gamma$ is negligible in the security parameter \label{informal-gamma}. 

Lastly, a scheduler process $\textsf{S}$ is responsible for playing the role of authority $\EC$ as well as scheduling all involved parties in a run according to defined protocol phases.

\subsubsection{Assumptions for Verifiability} \label{verifiability-assumptions}
\begin{enumerate}
    \item[P1] The public key encryption scheme $\mathcal{E}$ is correct, Pedersen Vector Commitment scheme is correct and computationally binding and all NIZKPs are correct and computationally sound. 
    \item[P2] The scheduler $\textsf{S}$, judge $\text{J}$ and $\PBB$ are honest: $\varphi = \mathsf{hon(S)} \wedge \mathsf{hon(J)}\wedge \mathsf{hon(\PBB)}$.
\end{enumerate}

\begin{theorem}[Verifiability, S1--S5]\label{verif-theorem}
Under the above assumptions, any run of \textsf{\proname} accepted by $\textsf{J}$ satisfies properties \textbf{S1--S5} except with probability negligible in the security parameter~$\ell$.
\end{theorem}

Verifiability follows from the properties of the underlying cryptographic primitives and the \PBB's enforcement rules.

\subsubsection{Security Proofs}
\begin{itemize}
    \item[S1] \textbf{Double-Vote Inhibition} is handled by $\PBB$ over blockchains through the derived property of persistence (\Cref{sec:pbb}).
    \item[S2] \textbf{Vote Immutability}: after the Voting phase ends, all vote-share commitments $c_i^{(j)}\ \forall j\in[n_t]$ are fixed on $\PBB$; immutability follows from the \emph{persistence} property of bulletin boards~\cite{bulletinboards} and the computational binding of Pedersen commitments~\cite{pedersoncommitment,pedersonvector}.
    \item[S3] \textbf{Cast-as-Intended} follows from the soundness of the sequential Audit-or-Cast protocol (\Cref{seq-audit-cast}): each honest audit detects a cheating tallier with probability $1/2$ via a Cut-and-Choose~\cite{cut-and-choose} argument, so $k$ independent audit rounds reduce the cheating probability to $2^{-k}$. Correctness of the voter's opening is additionally enforced by the validity tests (\Cref{validity-tests}).
    \item[S4] \textbf{Tally-as-Intended} follows from the computational binding and additive homomorphic properties of $\mathsf{CommVec}$ (\Cref{def:pedersenvector}): the aggregate commitment $c_\perp=\prod_i c_i$ opens to exactly $f_{\text{res}}(\mathbb{T})$ because re-randomization preserves the committed plaintext and the binding property rules out any equivocation.
    \item[S5] \textbf{End-to-End Verifiability}: any verifier can check $f_{\text{res}}(\mathbb{T})$ by verifying the NIZK $\pi_{\text{res}}$ published on $\PBB$; soundness of $\pi_{\text{res}}$ under the above assumptions ensures the declared tally is consistent with the committed votes.
\end{itemize}

\subsection{Privacy Framework}

We refer to the privacy framework as defined in \cite{kusters2011verifiability} for analyzing the \textit{privacy properties} of our protocol. We also refer to \cite{kryvos} for formal definitions of \textit{publicly tally-hiding}. While the privacy framework \cite{kusters2011verifiability} formalizes and measures coercion-resisitance as a part, we introduce our definitions of \textit{receipt-freeness} motivated by \cite{benaloisRF,kiayias2015e2e,a-critique-receipt-freeness}.

The main idea of \Cref{privacy-def} is showing the inability of the adversary to distinguish whether some voter $\mathsf{V}_{\text{obs}}$ who runs honest program voted for $m_0$ or $m_1$.

For a given voting method $(C,f_{\text{res}})$, voter $\vobs$ and $\vec{v}_{\text{obs}} \in C$, we consider runs of the protocol $P$ as $ (\hat{\pi}_{\mathsf{V}_{\text{obs}}}(m)||\pi^*||\pi_{\mathcal{A}})$ where $\hat{\pi}_{\mathsf{V}_{\text{obs}}}(m)$ is the honest process of voter $\vobs$ considering $m$ as their choice, $\pi^*$ is the composition of the processes of rest of the parties $P$ and $\pi_{\mathcal{A}}$ is adversary's process. 

\begin{definition}[Privacy]\label{privacy-def}
Let $P$ be the voting protocol, $\mathsf{V}_{\text{obs}}$ be the voter under observation and $\delta \in [0,1]$. Then, $P$ achieves $\delta$-privacy, if for all messages $m_0,m_1 \in C$ and all adversaries $\pi_{\mathcal{A}}$ the difference $$\mathsf{Pr}[(\hat{\pi}_{\mathsf{V}_{\text{obs}}}(m_0)||\pi^*||\pi_{\mathcal{A}})^{(\ell)} \mapsto 1] - \mathsf{Pr}[(\hat{\pi}_{\mathsf{V}_{\text{obs}}}(m_1)||\pi^*||\pi_{\mathcal{A}})^{(\ell)} \mapsto 1]$$ is $\delta$-bounded as a function of the security parameter $1^\ell$.\footnote{A function $f$ is $\delta$-bounded if $\forall c>0$, there exists $\ell_0$ such that $f(\ell)\leq\delta+\ell^{-c} \ \forall \ell>\ell_0$ \cite{replay-attacks-efficiency}}
\end{definition}

Thus, we aim to keep $\delta$ as small as possible since it specifies the adversary's advantage to ``break'' user privacy.

\subsubsection{Ideal Privacy}
Formal privacy results are formulated w.r.t. the privacy level $\delta^{\text{ideal}}_{(n_v,n_v^h,\mu)}(C,f_{\text{res}})$ of an ideal voting protocol $\mathcal{I}_{\text{voting}}(n_v,n_v^h,C,f_{\text{res}},\mu)$ for voting method $(C,f_{\text{res}})$ \cite[Figure 11]{kryvos}. 
Here, $n_v$ refers to the overall voters, $n_v^h$ refers to the honest voters in the run. Furthermore, we assume the honest voters to choose their votes from a distribution $\mu$ over choice-space $C$.

\subsubsection*{Assumptions for Privacy}\label{privacy-assumptions}

\begin{enumerate}
    \item The PKE scheme $\mathcal{E}$ is IND-CCA2-secure.
    \item The election commission $\EC$, public bulletin board $\PBB$ and $\Tdesign$ are honest.
    \item Communication channels between all parties ($\EC$, voters $v \in \mathcal{V}$, talliers $\tj \in \mathcal{T}$ and $\PBB$) are all \textit{mutually-authenticated}.\footnote{Note that mutual-authentication via mTLS\cite{mtls-rfc} is necessary owing to mutually adversarial nature of parties, esp. voters and talliers.}
    \item The messages from voters $v_i$ to talliers $\tj$ are private. (only required for strict receipt-freeness)
\end{enumerate}

\subsubsection{Publicly Tally-Hiding}
Intuitively, for a voting protocol $P$ to be publicly tally-hiding for some voting method $(C, f_{\text{res}}, \mu)$
\begin{itemize}
    \item \textbf{Public Privacy}: $P$ provides same privacy as the ideal voting protocol $\mathcal{I}_{\text{voting}}$ for voting method $(C, f_{\text{res}})$, assuming all talliers are honest.
    \item \textbf{Internal Privacy}: $P$ provides the same privacy as the ideal voting protocol $\mathcal{I}_{\text{voting}}$ for voting method $(C, f_{\text{complete}})$, where $f_{\text{complete}}$ returns the \textit{full tally}, assuming $t$-out-of-$n_t$ talliers are dishonest. In other words, talliers learn as much as they would've in non-tally-hiding voting protocols.
\end{itemize}

\begin{definition}[Publicly Tally-Hiding \cite{kryvos}]\label{publicly-tally-hiding-def}
Let $P$ be a voting protocol with a set of talliers $\mathcal{T}$ and $t \leq n_t$. We say that $P$ is $(\delta_P,\delta_i)$-publicly tally-hiding w.r.t $(\mathcal{T},t)$ iff:
\begin{enumerate}
    \item \textbf{Public Privacy $\delta_P$}: If all parties $\tj \in \mathcal{T}$ are honest, then $P$ achieves $\delta_P$-privacy.
    \item \textbf{Internal Privacy $\delta_i$}: If at most $t-1$ parties $\tj \in \mathcal{T}$ are dishonest, then $P$ achieves $\delta_i$-privacy.
\end{enumerate}
\end{definition}

The above theorem means that the public privacy level $\delta_P$ is the ideal one for $(C,f_{\text{res}},\mu)$ and its internal privacy $\delta_i$ is ideal one for $(C,f_{\text{complete}})$ as defined above.
The formal game-based proofs for \textit{internal privacy} and \textit{public privacy} \Cref{publicly-tally-hiding-theorem} are motivated from \cite{kryvos}.

\begin{theorem}[Internal Privacy \cite{kryvos}]\label{def:internal-privacy}
For all $m_0,m_1\in C$, for all programs $\pi^*$ of the remaining parties such that at least $n_v^h$ voters and at least one tallier are honest in $\pi^*$ (excluding $\vobs$), $$\mid\text{Pr}[(\hat{\pi}_{\vobs}(m_0)||\pi_{\adv})\mapsto1] -\text{Pr}[(\hat{\pi}_{\vobs}(m_1)||\pi_{\adv})\mapsto1] \mid$$ is $\delta_{(n_v,n_v^h,\mu)}^{\text{ideal}}(C,f_{\text{complete}})$-bounded as a function of security parameter $\ell$, where $f_{\text{complete}} $ return the complete tally.
\end{theorem}

\begin{theorem}[Public Privacy \cite{kryvos}]\label{public-privacy-thm}
For all $m_0,m_1\in C$, for all programs $\pi^*$ of the remaining parties such that at least $n_v^h$ voters and all talliers are honest in $\pi^*$ (excluding $\vobs$), $$\mid\text{Pr}[(\hat{\pi}_{\vobs}(m_0)||\pi_{\adv})\mapsto1] -\text{Pr}[(\hat{\pi}_{\vobs}(m_1)||\pi_{\adv})\mapsto1] \mid$$ is $\delta_{(n_v,n_v^h,\mu)}^{\text{ideal}}(C,f_{\text{res}})$-bounded as a function of security parameter $\ell$, where $f_{\text{res}} $ is the actual (tally-hiding) result function.
\end{theorem}

\begin{proof}[Proof Sketches, \Cref{def:internal-privacy} and  \Cref{public-privacy-thm}]
We replace \cite[Appendix I.2, Game 1]{kryvos} and \cite[Appendix I.3, Game 1]{kryvos} keeping the proof strategy similar for our protocol.

\begin{itemize}
        \item[\textbf{Game 1}] In $\hat{\pi}_{\mathsf{H}}^{(1)}(m)$, we modify the tallier $\mathsf{T}_1$ for \Cref{def:internal-privacy}/$\mathsf{T}$ for \Cref{public-privacy-thm} (assumed honest) to abort if either message decryption over the authenticated channel between any arbitrary honest voter $v_i$ and tallier or the digital signature of the voter-share commitment $c_i^{(1)}$ received fails.  
    Due to the correctness of the encryption scheme $\mathcal{E}$ and the digital signature scheme, Games 0 and 1 are perfectly indistinguishable. Since tallier is honest, it opens the re-randomization of the received ciphertext $c_i^{(1)}$ after $\audit$ phase correctly on $\PBB$.
\end{itemize}

\end{proof}

\begin{theorem}[Publicly Tally-Hiding \cite{kryvos}]\label{publicly-tally-hiding-theorem}
    Let $\mathcal{T}$ as defined in \Cref{tab:terminology} and $t=n_t$. Then, assuming \Cref{privacy-assumptions} hold, the voting protocol $P_{\mathsf{\proname}}(n_v,n_t,C,f_{\text{res}},\mu)$ is $(\delta^{\text{ideal}}_{(n_v,n_v^h,\mu)}(C,f_{\text{res}}),\delta^{\text{ideal}}_{(n_v,n_v^h,\mu) }(C,f_{\text{complete}}))$-publicly tally-hiding w.r.t $(\mathcal{T},n_t)$.
\end{theorem}

\subsubsection{Receipt-Freeness}

Our motivation for receipt-freeness comes directly from \cite{benaloh-challenge,goos_receipt-free_1998}.

\begin{definition}[Receipt-Freeness]\label{receipt-freeness-protocol-defn}
    We say that a protocol $P$ is \textbf{receipt-free} if there exists a simulator $\mathcal{S}$ such that for any voter $v$ and any adversary $\mathcal{A}$ (who corrupts $v$ after the voting phase and demands all secrets): $$\mathsf{View}_{\adv}^{\text{Real}}(v,r,\PBB) \approx \mathsf{View}_{\adv}^{\text{Simulated}}(v',r',\PBB)$$ where in \textbf{Real view}, $\adv$ sees the $\PBB$ (containing $\tilde{c}_v$ and all voter's secrets $(v,r)$ while in \textbf{Simulated view}, $\adv$ sees $\PBB$ but voter claims they voted for $v'\neq v$ with secrets $(v',r')$ (amongst all other secrets).
\end{definition}

\begin{theorem}[Receipt-Freeness for $\mathsf{\proname}$]\label{receipt-freeness-protocol-thm}
The $\mathsf{\proname}$ protocol is \textbf{receipt-free} (\Cref{receipt-freeness-protocol-defn}) against an adversary $\mathcal{A}$ who controls the network (except secure channels to talliers) and can coerce the voter to reveal their private state $(v, r)$, assuming at least one tallier is honest.
\end{theorem}
\begin{proof}[Proof Sketch]
\begin{enumerate}
    \item The coercer observes the PBB entry $\tilde{c}_i^{(j)}$ and receives the voter's receipt $c_i^{(j)}$.
    \item The coercer tries to verify that $\tilde{c}_i^{(j)}$ corresponds to $c_i^{(j)}$.
    The coercer calculates $c_{calc} = \mathsf{CommVec}(c_i^{(j)}; \tilde{r}_i^{(j)})$. The coercer observes $\tilde{c}_i^{(j)}$ on the $\PBB$. The validity condition is $\tilde{c}_i^{(j)} = c_{calc} \cdot h^{\tilde{r}}$.
    \item The value $\tilde{r}$ is generated by the Tallier. In the Voting phase, $\tilde{r}$ is never sent to the voter assuming at least one honest tallier.
    \item \textbf{Indistinguishability}: Since $\tilde{r}$ is drawn uniformly from $\mathbb{Z}_q$, the term $h^{\tilde{r}}$ is a uniform random element in $\mathbb{G}$. Therefore, $\tilde{c}_i^{(j)}$ is uniformly distributed in $\mathbb{G}$, independent of $c_{calc}$. Thus, for any other fake vote $v'$ and fake randomness $r'$, there exists a theoretical $\tilde{r}'$ such that $\tilde{c}_i^{(j)} = \mathsf{CommVec}(v'; r') \cdot h^{\tilde{r}'}$.
\end{enumerate}
    The coercer does not know $\tilde{r}$, therefore they cannot determine if $\tilde{c}_i^{(j)}$ was derived from the voter's real receipt $(v, r)$ or a fake receipt $(v', r')$. The PBB entry $\tilde{c}_i^{(j)}$ effectively becomes a ``perfectly hiding'' commitment \cite{pedersoncommitment} of the vote relative to the coercer, even if the coercer knows the input randomness.
\end{proof}

\subsubsection{Privacy Proofs}

\begin{itemize}
    \item[P1] \textbf{Public Privacy} follows directly from \Cref{public-privacy-thm}.
    \item[P2] \textbf{Publicly Tally-Hiding} follows directly from \Cref{publicly-tally-hiding-theorem}.
    \item[P3] \textbf{Receipt-Freeness} follows directly from \Cref{receipt-freeness-protocol-thm}.
\end{itemize}

\section{Conclusion and Future Work}

We introduce \textsf{\proname}, a voting protocol that achieves public auditability, tally-hiding, and receipt-freeness within a single cryptographic framework. By combining an Audit-or-Cast mechanism with tallier-side re-randomization, the protocol enforces cast-as-intended without relying on trusted voting devices and prevents the construction of transferable receipts.

Our results demonstrate that explicit audit challenges provide a principled alternative to trusted-client assumptions in electronic elections. Future work includes Batch-NIZKs for lower voter cost, supporting richer voting rules, and evaluating usability at scale.

\bibliographystyle{IEEEtran}
\bibliography{bibliography}

\end{document}